\documentclass[12pt,a4paper]{article}
\usepackage{amsfonts,color,comment}
\usepackage{amsthm, thmtools, amsmath, amscd, amssymb}
\usepackage[utf8]{inputenc}
\usepackage{t1enc}
\usepackage[mathscr]{eucal}
\usepackage{indentfirst}
\usepackage{graphicx}
\usepackage{graphics}
\usepackage{pict2e}
\usepackage{epic}
\usepackage{url}
\usepackage{epstopdf}
\usepackage[english]{babel}
\usepackage{hyperref}
\usepackage{pgf,tikz}
\usetikzlibrary{arrows}
\definecolor{zzttqq}{rgb}{0,0,0.1}
\definecolor{zzttaa}{rgb}{0.6,0.2,0}
\usepackage{hhline}
\usepackage{enumitem}
\usepackage{tikz}
\usetikzlibrary{fit}
\usetikzlibrary{calc, shapes, positioning}
\usetikzlibrary{patterns.meta}
\usepackage{mathrsfs,pgf}
\usepackage{amsthm}
\usepackage{soul}
\usepackage{tikz,amsmath,helvet}
\usetikzlibrary{arrows.meta, decorations.pathreplacing}
\newcommand\scalemath[2]{\scalebox{#1}{\mbox{\ensuremath{\displaystyle #2}}}}
\usepackage[margin=2.5cm]{geometry}

\usepackage{fancyhdr}

\theoremstyle{definition}

\theoremstyle{plain}
\newtheorem{Th}{Theorem}
\newtheorem{lemma}[Th]{Lemma}

\theoremstyle{remark}

 \theoremstyle{definition}

\newcommand{\ZZ}{{\mathbb Z}}

\def\ZZ{\mathbb{Z}}

\def\HH{\mathcal{H}}

\def\FF{\mathcal{F}}
\def\AA{\mathcal{A}}
\def\PP{\mathcal{P}}
\def\OO{\mathcal{O}}

\mathcode`l="8000
\begingroup
\makeatletter
\lccode`\~=`\l
\DeclareMathSymbol{\lsb@l}{\mathalpha}{letters}{`l}
\lowercase{\gdef~{\ifnum\the\mathgroup=\m@ne \ell \else \lsb@l \fi}}%
\endgroup

\title{The Six-Dimensional Tarski Problem is $\Omega(\log^3n)$-Hard}

\author{
Krist\'of Z\'olomy\thanks{ELTE Eötvös Loránd University, Budapest. Supported by the EXCELLENCE-24 grant no.~151504 Combinatorics and Geometry of the National Research, Development and Innovation Office.
}
}

\makeatletter

\begin{document}

\thispagestyle{empty}
\date{October 15, 2026}
\maketitle

\begin{abstract}
We study the deterministic query complexity of finding a fixed point of a monotone function $f:[n]^d\to[n]^d$, where querying $x$ returns $f(x)$.
For fixed $d\ge2$, earlier work established a lower bound of $\Omega((\log n)^2)$ and an upper bound of $O_d((\log n)^{\lceil(d+2)/3\rceil})$ for the query complexity.
Very recently, several papers obtained improved lower bounds using AI-assisted proofs.
In this paper, we further improve these lower bounds to $\text{Tarski}(n,d)=\Omega_d((\log n)^{2+\lfloor\log_3(d/2)\rfloor})$.
In particular, we obtain the optimal $\Theta((\log n)^3)$ deterministic query complexity in six dimensions.
\end{abstract}

\section{Introduction}

A function $f:[n]^d\to[n]^d$ is monotone if $x\leq y$ implies $f(x)\leq f(y)$, where inequalities are interpreted coordinatewise.
Tarski's fixed point theorem~\cite{Tar55} guarantees that every such function has a fixed point, that is, a point $x$ with $f(x)=x$.
Finding pure Nash equilibria in supermodular games and solving parity, mean-payoff, and simple stochastic games can all be reduced to finding such a fixed point~\cite{EPRY20}.
We study the deterministic query complexity of this problem: each query at $x$ returns the entire vector $f(x)$, and computation between queries is unrestricted.
We write $\text{Tarski}(n,d)$ for the minimum worst-case number of queries needed to find a fixed point.

Recursive binary search gives an $O_d((\log n)^d)$ upper bound, as shown by Dang, Qi, and Ye~\cite{DQY20}.
Etessami, Papadimitriou, Rubinstein, and Yannakakis~\cite{EPRY20} proved an $\Omega((\log n)^2)$ lower bound in two dimensions, even for randomized algorithms, settling the two-dimensional case.
Fearnley, P{\'a}lv{\"o}lgyi, and Savani~\cite{FPS22} obtained an $O((\log n)^2)$-query algorithm in three dimensions and a decomposition theorem that together yield an $O_d((\log n)^{\lceil 2d/3\rceil})$ upper bound in general.
In 2022, Chen and Li~\cite{CL22} improved this to $O_d((\log n)^{\lceil(d+1)/2\rceil})$.
In 2025, Haslebacher and Lill~\cite{HL26} gave a substantially different optimal three-dimensional algorithm based on the \emph{level sets} of the hypercube.
In March 2026, Chen, Li, and Yannakakis~\cite{CLY26} obtained an optimal, $O((\log n)^2)$-query algorithm in four dimensions and an $O_d((\log n)^{\lceil(d+2)/3\rceil})$ upper bound in general.

In October 2026, Chen, Gao, Li, Rubinstein, and Yannakakis~\cite{CGLRY26} reduced the upper bound to $O(5^d(\log n)^{\lceil\log_2 d\rceil})$ for $d\geq3$, and proved lower bounds whose exponent in $\log n$ also grows logarithmically with $d$, even for randomized algorithms.
Their construction gives an $\Omega((\log n)^3)$ lower bound in ten dimensions.
Independently, Liu and Wang~\cite{LW26} developed a recursive construction giving an $\Omega((\log n)^3)$ deterministic lower bound in nine dimensions and increasingly large logarithmic exponents in higher dimensions.

Li, Ma, Yang, and Zhao~\cite{LMYZ26} subsequently established bootstrapping theorems for quantum query lower bounds, obtaining, in particular, an $\Omega((\log n)^3)$ lower bound in eight dimensions.
All three works~\cite{CGLRY26,LW26,LMYZ26} establish general lower bounds of the form $\Omega_d((\log n)^{\log_4 d+O(1)})$, with differences in the additive constant in the exponent and the computational models covered.
For further background and applications, see~\cite{CGLRY26}.

Our main result is the following.

\begin{Th}\label{thm:main}
For every fixed integer $d\geq2$, the deterministic query complexity
of finding a Tarski fixed point satisfies
\begin{equation*}
        \text{Tarski}(n,d)=\Omega_d\left((\log n)^{2+\lfloor\log_3(d/2)\rfloor}\right).
\end{equation*}

In particular, $\text{Tarski}(n,6)=\Omega((\log n)^3)$,  matching the $O((\log n)^3)$ upper bound of~\cite{CGLRY26} and establishing the optimal deterministic query complexity in this dimension.
\end{Th}

We expect our construction to yield a similar lower bound for randomized algorithms, but we do not establish this extension here; our result concerns only deterministic query complexity.
The proof uses a recursive construction starting in two dimensions:
each step triples the dimension and increases the exponent of
$\log n$ by one.

\section{Preliminaries}

Throughout, intervals contain only integers, and $[n]=[1,n]$. Fix an integer $k \geq 0$. 
In our proof, we will discuss $2 \cdot 3^r$-dimensional Tarski instances on cubes of side length $n_r$ with a parameter $m_r$, defined recursively as follows. Pick an integer $m_0$ and set $n_0=m_0(m_0+1)$. For $r\geq1$, define   $m_r=\left\lfloor\sqrt{\frac{m_{r-1}}{8}}\right\rfloor$, and  $n_r=(8m_r^2+4m_r+1)n_{r-1}$.
Choose $m_0$ sufficiently large that $m_r \geq 2$ for every $r \leq k$. For fixed $k$, these sequences satisfy $\log m_k  = \Theta_k(\log n_k)$.

Let $\AA_k$ be any algorithm resolving the $2 \cdot 3^k$-dimensional Tarski search problem.
We will recursively define a family $\FF_k$ and a \emph{simulator} $\Sigma_k$ for every dimension $d=2 \cdot 3^k$ for $k \geq 0$,  such that:

\begin{itemize}
    \item $\FF_k$ is the set of possible hidden functions, its elements are monotone functions $f:[n_k]^d \to [n_k]^d$ with a unique fixed point. To each such function, we assign an integer label $J(f) \in [1, m_k]$, that is uniquely determined by this fixed point.
    \item Any queries of $\AA_k$ about $f$  will be answered by $\Sigma_k$.
    \item $\Sigma_k$ has access to $f$ only through \emph{local queries}, i.e. queries of the form $J(f) \overset{?}{\leq} t$, answered by the oracle $\OO_k$.
    \item At any point of the algorithm, if $\AA_k$ has made  $Q$ queries so far, $\Sigma_k$ will have made at most $O_k\left(\frac{Q}{\log^{k+1} {n_k}}\right)$ local queries.
    \item Any value of $J$ that is consistent with all previous local queries is the label of some $f \in \FF_k$ that is consistent with the answers supplied by $\Sigma_k$ so far. Consequently, any algorithm locating the fixed point of a hidden function must also determine its label.
\end{itemize}

\begin{lemma}\label{lem:simulatorbound}
    If a pair consisting of a  function family and a simulator satisfying the above properties exists for any $0 \leq k' \leq k$, then $\text{Tarski}(n_k,2\cdot 3^k) = \Omega(\log^{k+2} n_k)$.
\end{lemma}

\begin{proof}
    We maintain an interval $L$ of labels consistent with all local comparison answers, initially $L=[1,m_k]$. We forward every evaluation query of $\AA_k$ to $\Sigma_k$. Whenever $\Sigma_k$ makes a local comparison query, we answer it by retaining the larger of the two resulting subintervals of $L$. Thus each local comparison retains at least half of the remaining labels.

    By the consistency property, each label in $L$ admits a function in $\FF_k$ that is consistent with all evaluation answers so far. Functions with different labels have different unique fixed points, so $\AA_k$ cannot determine the fixed point while $|L|>1$. Hence $\Sigma_k$ must make at least $\log m_k$ local queries. By the query count property, $\AA_k$ must therefore make $\Omega_k(\log m_k\cdot\log^{k+1}n_k)=\Omega_k(\log^{k+2}n_k)$ evaluation queries.
\end{proof}

In the rest of the paper, we will discuss the construction of families $\FF_k$ and the simulators $\Sigma_k$ satisfying the properties defined above. By Lemma \ref{lem:simulatorbound}, this suffices for the proof of Theorem \ref{thm:main}.

Although we construct instances only for side lengths $n_k$, for each fixed $k$ and every sufficiently large $n$ we can choose $m_0$ so that $n_k\leq n$ and $n_k=\Theta_k(n)$; padding these instances to $[n]^{2\cdot3^k}$ therefore yields the same asymptotic lower bound for arbitrary $n$.
For a general fixed dimension $d\geq2$, take $k=\lfloor\log_3(d/2)\rfloor$ and pad the remaining coordinates by setting their output values to $1$.

\section{The two-dimensional base case}

In this section, we will describe the family $\FF_0$ and the construction of $\Sigma_0$. We divide the interval $[1,n_0]$ into $m_0+1$ subintervals $I_i^{(0)} =  [(i-1)m_0+1, im_0]$.
On each subinterval, we define a family of monotone functions, where each member has a unique fixed point $p_i$ (also referred to as a \emph{local fixed point}) and is assigned a bit $e_i \in \{0,1\}$. For each such sequence of fixed point-bit pairs $(p_i,e_i)_{i=1}^{m_0+1}$ such that the sequence $(e_i)$ is nonincreasing, with $e_1=1$ and $e_{m_0+1}=0$, we will define a function $f \in \FF_0$.

Let $ \mathcal{P} = \{P=(p_1, p_2, \dots, p_{m_0+1}) : \forall i \;p_i \in I_i^{(0)}\} $ be the set of possible local fixed point sequences. 
First, we define  functions  $h_{P,j,i}: [0,n_0] \to I_i^{(0)} $ and $e_{P,j,i}: [0,n_0] \to \{0,1\}$ for $i \in [1,m_0+1]$, $j \in [1,m_0]$ and $P \in \mathcal{P}$: 

\begin{equation*}
   (h_{P,j,i}(x), e_{P,j,i}(x) ) = \begin{cases}
    ((i-1)m_0+1,0) & \text{if } x < (i-1)m_0+1, \\
    (x+1,0) & \text{if }  (i-1)m_0+1 \leq x < p_i, \\
    (x, 1_{i \leq j} ) & \text{if }   x = p_i, \\
    (x-1,1) & \text{if }  p_i < x \leq im_0 , \\
    (im_0,1) & \text{if }  im_0 <  x. \\
\end{cases} 
\end{equation*}

For an example, see Figure \ref{fig:base-case-local}.

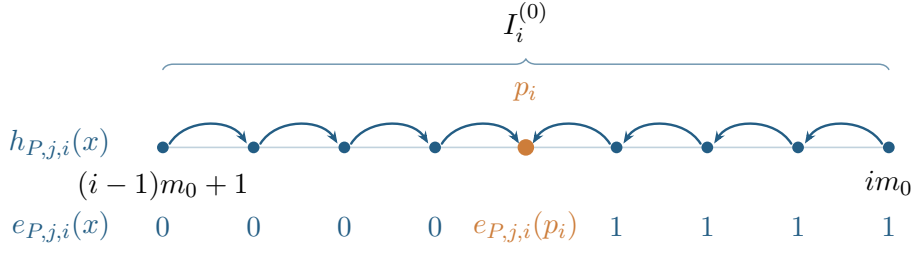
\begin{figure}[tbp]
    \centering
    \begingroup
\definecolor{baseLocalBlue}{RGB}{35,93,132}
\definecolor{baseLocalOrange}{RGB}{205,125,58}
\begin{tikzpicture}[
  x=1cm,y=1cm,
  font=\small,
  >=Stealth,
  line cap=round,
  line join=round,
  local map/.style={draw=baseLocalBlue,line width=.9pt,-{Stealth[length=1.65mm,width=1.25mm]}},
  value/.style={text=baseLocalBlue,inner sep=0pt},
  every node/.style={fill=none}
]
  \draw[draw=baseLocalBlue!28,line width=.6pt] (0,0)--(9.6,0);

  \foreach \q in {0,1,2,3} {
    \pgfmathsetmacro{\LocalStart}{1.2*\q+.09}
    \pgfmathsetmacro{\LocalEnd}{1.2*(\q+1)-.09}
    \draw[local map]
      (\LocalStart,.05)
      to[out=58,in=122]
      (\LocalEnd,.05);
  }
  \foreach \q in {5,6,7,8} {
    \pgfmathsetmacro{\LocalStart}{1.2*\q-.09}
    \pgfmathsetmacro{\LocalEnd}{1.2*(\q-1)+.09}
    \draw[local map]
      (\LocalStart,.05)
      to[out=122,in=58]
      (\LocalEnd,.05);
  }
  \foreach \q in {0,1,2,3,5,6,7,8}
    \fill[baseLocalBlue] ({1.2*\q},0) circle (2.2pt);
  \fill[baseLocalOrange] (4.8,0) circle (3.1pt);

  \node[anchor=east,text=baseLocalBlue] at (-.55,.05)
    {$h_{P,j,i}(x)$};
  \node[text=baseLocalOrange,anchor=south] at (4.8,.47) {$p_i$};
  \node[anchor=north] at (0,-.17) {$(i-1)m_0+1$};
  \node[anchor=north] at (9.6,-.17) {$im_0$};
  \draw[draw=baseLocalBlue!65,line width=.55pt,
        decorate,decoration={brace,amplitude=4pt}]
    (0,1.03)--(9.6,1.03)
    node[midway,above=6pt,text=black] {$I_i^{(0)}$};

  \node[anchor=east,text=baseLocalBlue] at (-.55,-1.05)
    {$e_{P,j,i}(x)$};
  \foreach \q in {0,1,2,3}
    \node[value] at ({1.2*\q},-1.05) {$0$};
  \node[text=baseLocalOrange,inner sep=0pt] at (4.8,-1.05)
    {$e_{P,j,i}(p_i)$};
  \foreach \q in {5,6,7,8}
    \node[value] at ({1.2*\q},-1.05) {$1$};
\end{tikzpicture}
\endgroup
    \caption{The maps $h_{P,j,i}$ and $e_{P,j,i}$ on $I_i^{(0)}$.
    The map $h_{P,j,i}$ fixes $p_i$ and moves every other point one step
    towards it. Only $e_{P,j,i}(p_i)=\mathbf{1}_{i\leq j}$ depends on $j$.}
    \label{fig:base-case-local}
\end{figure}

For an element $(x,y) \in [1,n_0]^2$, let $x= (x_1-1)m_0 +x_0$ and $y=(y_1-1)m_0+y_0$ for some integers $1 \leq x_1, y_1 \leq m_0+1$ and $1 \leq x_0, y_0 \leq m_0 $. 

For $j \in [1,m_0]$ and $P \in \mathcal{P}$, define $f_{P,j}:[1,n_0]^2 \to [1,n_0]^2$ as follows. We first define a target interval for $x$ and $y$, which will contain their coordinatewise images. The index of the target interval for $x$ is given by a function of $y$: $T_x(y) = \max\{1, y_1 -1+e_{P,j, y_1 }(y) \} $, and similarly for $y$, the target interval index is $T_y(x) = \min\{m_0+1, x_1+ e_{P,j,x_1}(x)\}$. Now, let
\begin{equation*}
    f_{P,j}(x,y)  = \left( h_{P,j,T_x(y)} (x), h_{P,j,T_y(x)} (y) \right),
\end{equation*}

and $\FF_0= \{f_{P,j} : P \in \mathcal{P}, j \in [1,m_0] \}$, with $J(f_{P,j}) = j$. For an example, see Figure \ref{fig:base-case-fixed-point}.

\begin{figure}[tbp]
    \centering
    \begingroup
\definecolor{baseBlue}{RGB}{35,93,132}
\definecolor{baseOrange}{RGB}{205,125,58}
\definecolor{baseGrid}{RGB}{212,222,229}
\begin{tikzpicture}[
  x=.75cm,y=.75cm,font=\small,
  line cap=round,line join=round,
  every node/.style={fill=none},
  axis/.style={draw=black!75,line width=.7pt,-{Stealth[length=1.8mm]}},
  projection/.style={draw=baseOrange!85!black,densely dashed,line width=.75pt},
  interval brace/.style={draw=black!55,line width=.45pt,decorate,
    decoration={brace,amplitude=2.2pt}}
]

  \foreach \i in {1,...,7} {
    \ifnum\i<4
      \path[fill=baseBlue!12] (\i-1,4) rectangle (\i,5);
    \else
      \path[fill=baseOrange!12] (\i-1,4) rectangle (\i,5);
    \fi
    \draw[draw=baseGrid,line width=.6pt] (\i-1,4) rectangle (\i,5);
    \node at (\i-.5,4.5) {$I_{\i}^{(0)}$};
    \ifnum\i<4
      \node[text=baseBlue] at (\i-.5,3.25) {$1$};
    \else
      \node[text=baseBlue] at (\i-.5,3.25) {$0$};
    \fi
  }
  \draw[baseOrange,line width=1.15pt] (2,4) rectangle (4,5);

  \begin{scope}[shift={(9.4,.6)}]
    \path[fill=baseOrange!22] (2,3) rectangle (3,4);
    \draw[draw=baseGrid,line width=.45pt,step=.75cm] (0,0) grid (7,7);
    \draw[draw=baseOrange,line width=1pt] (2,3) rectangle (3,4);
    \draw[axis] (0,0)--(7.55,0) node[anchor=west] {$x$};
    \draw[axis] (0,0)--(0,7.55) node[anchor=south] {$y$};

    \foreach \i in {1,...,7} {
        (\i-1,-.23)--(\i,-.23);
      \node[anchor=north] at (\i-.5,-.25) {$I_{\i}^{(0)}$};
      \node[anchor=east] at (-.25,\i-.5) {$I_{\i}^{(0)}$};
    }

    \draw[projection] (2.4,0)--(2.4,3.6)--(0,3.6);
    \draw[baseOrange!85!black,line width=.8pt] (2.4,-.07)--(2.4,.07);
    \draw[baseOrange!85!black,line width=.8pt] (-.07,3.6)--(.07,3.6);
    \node[anchor=south west,text=baseOrange!85!black] at (2.42,.12) {$p_3$};
    \node[anchor=south west,text=baseOrange!85!black] at (.12,3.67) {$p_4$};
    \fill[baseOrange!85!black] (2.4,3.6) circle (2.8pt);
    \node[anchor=south,font=\normalsize] at (3.7,7.85)
      {$f_{P,3}(p_3,p_4)=(p_3,p_4)$};
  \end{scope}
\end{tikzpicture}
\endgroup
    \caption{The base-family example with $m_0=6$ and $j=3$.
    The values $e_{P,3,i}(p_i)$ change from $1$ to $0$ between
    $I_3^{(0)}$ and $I_4^{(0)}$.  The highlighted cell is $I_3^{(0)}\times I_4^{(0)}$;
    the marked point $(p_3,p_4)$ is the unique fixed point of $f_{P,3}$.}
    \label{fig:base-case-fixed-point}
\end{figure}
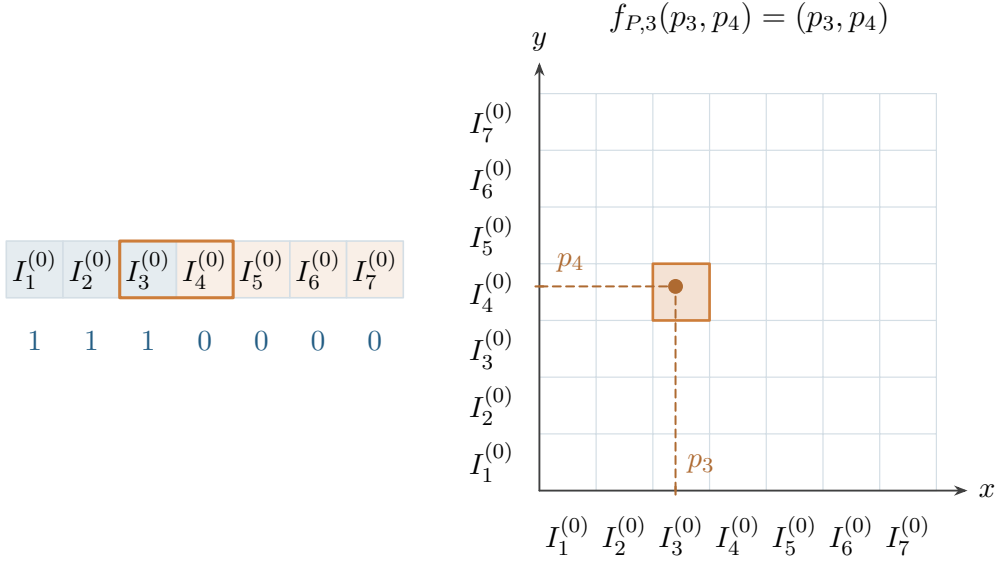

\begin{lemma}\label{lem:fmonuniq0}
    For each $(P,j)$, the function $f_{P, j}$ is monotone, and has a unique fixed point  $ (p_j, p_{j+1}) \in I_{j}^{(0)} \times I_{j+1}^{(0)}$.
\end{lemma}
\begin{proof}
    For fixed $j, P$ and $x$, the function $y \mapsto  h_{P,j,T_x(y)} (x)$ is monotone, and for fixed $y$, the map $x \mapsto  h_{P,j,T_x(y)} (x)$ is also monotone by construction. Similarly, the second coordinate is also monotone in $(x,y)$, which shows the monotonicity of $f_{P, j}$.
    If $(x,y)$ is a fixed point, the target interval of each coordinate must equal its starting interval, i.e. 
    \begin{equation}\label{eq:interval}
        x_1 = \max\{1,y_1-1+e_{P,j, y_1 }(y)\} \text{ and } y_1 = \min\{m_0+1 ,x_1 + e_{P,j,x_1} (x)\}.
    \end{equation}
  Additionally, both $x$ and $y$ must be fixed points of their corresponding functions: $h_{P,j,x_1} (x) = x = p_{x_1}$ and $ h_{P,j,y_1} (y) = y = p_{y_1}$, therefore $e_{P,j, x_1 }(x) = 1_{x_1 \leq j}$ and $e_{P,j, y_1 }(y) = 1_{y_1 \leq j}$ by definition. Substituting in (\ref{eq:interval}) gives $x_1 = \max\{1,y_1-1+1_{y_1 \leq j}\} \text{ and } y_1 = \min\{m_0+1 ,x_1 + 1_{x_1 \leq j}\}.$
   If $x_1 >j$, the second equation gives $y_1=x_1$, implying $y_1=  \max \{ 1, y_1-1\}$, thus $1=y_1=x_1>j$, a contradiction. 
   If $x_1 \leq j< m_0+1$, the second equation gives $y_1=x_1+1=1+  \max\{1,y_1-1+1_{y_1 \leq j}\}$, thus $y_1 >j$, implying $x_1=j, y_1=j+1$. Conversely, direct substitution shows that \((p_j,p_{j+1})\) is indeed a fixed point.
\end{proof}

\smallskip

We will now describe the simulator $\Sigma_0$. As a reminder, $\Sigma_0$ takes queries from an algorithm $\AA_0$ and answers them without having full access to the hidden function $f$. The only information it receives consists of answers to queries of the form $J(f) \overset{?}{\leq} t$, where $J(f)$ is the label of $f$, i.e. the index of the interval containing the smaller coordinate of its fixed point, and $t$ is an arbitrary integer. 
It guarantees that for each possible label $J$ of the hidden function consistent with its comparison queries, there is a function $f_J$ with label $J$ that is consistent with all answers given to $\AA_0$ so far. Additionally, at any point of the algorithm, if $\AA_0$ has made $Q$ queries to $\Sigma_0$, the number of queries made by  $\Sigma_0$ is at most $O\left(\frac{Q}{\log n_0} \right)$.

Note that for $\Sigma_0$, $j=J(f)$ is an unknown integer, while throughout the algorithm, it maintains a subset of $\PP$ that is consistent with all answers given to $\AA_0$. 

First, observe that each query $(x,y)$ about a function $f_{P,j}$ can be answered using at most four evaluations of pairs $(h_{P,j,i}(u), e_{P,j,i}(u))$: first evaluate $e_{P,j,x_1}(x)$ and $e_{P,j,y_1}(y)$, which allows us to calculate $T_x(y)$ and $T_y(x)$. Now it is sufficient to query $h_{P,j,i}$ with $i=T_x(y)$, and $T_y(x)$, respectively.

Therefore, it is sufficient for $\Sigma_0$ to answer queries of the form $(h_{P,j,i}(x), e_{P,j,i}(x) )$. We now describe the strategy of $\Sigma_0$ for answering these queries.
For each interval $I_i^{(0)}$, $\Sigma_0$ maintains a subinterval $\mathcal{I}_i^{(Q)}$ containing the possible locations of $p_i$ after $Q$ pair evaluations, in other words, the set of possible local fixed point sequences after step $Q$ is  $\mathcal{I}_1^{(Q)} \times \mathcal{I}_2^{(Q)} \times \dots \times\mathcal{I}_{m_0+1}^{(Q)}$. 
Initially, $\mathcal{I}_i^{(0)}=I_i^{(0)}$ for each $i$, and after a query $(h_{P,j,i}(x), e_{P,j,i}(x) )$, $\mathcal{I}_{i'}^{(Q+1)} = \mathcal{I}_{i'}^{(Q)}$ for all $i' \neq i$.

In a general step, consider a query $(h_{P,j,i}(u), e_{P,j,i}(u) )$. If $u \notin\mathcal{I}_i^{(Q)}$, the answer is automatic by the promised property of $h$ and $e$. Otherwise, if $|\mathcal{I}_i^{(Q)}| \geq 2$, deleting $u$ divides $\mathcal{I}_i^{(Q)}$ into two subintervals $\mathcal{I}_i^{(Q)-}$ and $\mathcal{I}_i^{(Q)+}$ (at least one of which is nonempty); $\Sigma_0$ chooses the longer one as $\mathcal{I}_i^{(Q+1)}$ (breaking ties arbitrarily). It answers $(u-1,1)$ if it retained  $\mathcal{I}_i^{(Q)-}$, otherwise it answers $(u+1,0)$.

If $|\mathcal{I}_i^{(Q)}| = 1$, then $\mathcal{I}_i^{(Q)}=\{p_i\}$, thus the answer is $\left(p_i,1_{i\leq j}\right)$. In this case, $\Sigma_0$ makes the query $j \overset{?}{\leq } i-1$, unless its answer is already known. If the answer is Yes, it answers $(p_i,0)$, otherwise it answers $(p_i,1)$.

\smallskip

This strategy maintains the required consistency property: after step $Q$, fix some $j$ that is consistent with all queries of $\Sigma_0$ made so far, and choose $p_i \in \mathcal{I}_i^{(Q)}$ arbitrarily. 
This uniquely determines $(h_{P,j,i}(u), e_{P,j,i}(u) )$ for each $u \in [n_0]$, and therefore it also determines $f$, a monotone function with label $j$ that is consistent with all answers so far.

It remains to bound the number of queries made by $\Sigma_0$ after $Q$ queries made by $\AA_0$. 
The simulator $\Sigma_0$ only makes queries when a subinterval $\mathcal{I}_i^{(Q')}$ of an interval $I_i^{(0)}$ reaches cardinality $1$. A query $(h_{P,j,i}(u), e_{P,j,i}(u) )$ at step $Q'$  leaves all intervals $\mathcal{I}_{i'}^{(Q')}$ for $i' \neq i$ unchanged, while $|\mathcal{I}_i^{(Q'+1)} | \geq \lfloor|\mathcal{I}_i^{(Q')} |/2 \rfloor$. This implies that an interval reaching cardinality $1$ has received $\Omega(\log m_0) $ queries and causes at most one new local query. Since each query of $\AA_0$ causes at most four pair evaluations and $\log m_0 = \Theta(\log n_0)$, this yields the desired upper bound on the number of queries made by $\Sigma_0$.

\section{The general case}

\subsection{The hidden function family $\FF_k$}

For $k \geq 1$, we will define the family $\FF_k$  using the definition and properties of the function family $\FF_{k-1}$ of the previous level.
The functions $F \in \FF_k$ map $[1, n_k]^d$ to itself, where $d = 2 \cdot 3^k$.

First, we group the variables of the input into three groups of $d/3 = 2\cdot 3^{k-1}$: the values of an  element $F\in \FF_k$ will be defined as a function of $X,Y,Z$, where $X, Y,$ and $Z$ consist of the first, second, and third $d/3$ coordinates of the input, respectively. Denote the $d$-dimensional all-ones vector by $\textbf{1}_d$.

Divide $[1,n_k]$ into the intervals $I_i^{(k)}=[(i-1)n_{k-1}+1,in_{k-1}]$ for $1\leq i\leq n_k/n_{k-1}$, and choose functions $G_1',G_2',\dots,G_{n_k/n_{k-1}}'\in\FF_{k-1}$ in a way we will describe later. Define $C_i=(I_i^{(k)})^{d/3}$, and for an interval $I=[a,b]$, define the function $\text{clip}_I: \ZZ \to \ZZ, \text{clip}_I(x) = \min\{ \max\{x, a \},b\} $. 
In dimension $D$, for a product of intervals $[a_1, b_1] \times [a_2,b_2] \times \dots \times [a_D,b_D]$, define the clipping function $\text{clip}_{[a_1, b_1] \times [a_2,b_2] \times \dots \times [a_D,b_D]} : \ZZ^D \to \ZZ^D$ coordinatewise:
\begin{equation*}
    \text{clip}_{[a_1, b_1] \times [a_2,b_2] \times \dots \times [a_D,b_D]}(x) = \left(\text{clip}_{[a_1,b_1]} (x_1), \text{clip}_{[a_2,b_2]} (x_2), \dots, \text{clip}_{[a_D, b_D]} (x_D) \right).
\end{equation*}

The clipping function maps a vector to the closest point of the given hyperrectangle.
We translate the domain and range of each $G_i'$ to $C_i$ and extend its domain to $[1,n_k]^{d/3}$ by clipping: for $x\in[1,n_k]^{d/3}$ and $i\in[1,n_k/n_{k-1}]$, let

\begin{multline}
    \label{eq:Gdef}
    G_i(x) = G_i'( \text{clip}_{C_i} (x) - (i-1)n_{k-1} \textbf{1}_{d/3} ) + (i-1)n_{k-1} \textbf{1}_{d/3}, \text{ and let } \\ G_0=G_1, G_{n_k/n_{k-1}+1} = G_{n_k/n_{k-1}}.
\end{multline}

The modified functions keep their labels, i.e. $J(G_i)=J(G_i')$ for all $i$.
Now, an element $F \in \FF_k$ will have the form 
 \begin{equation*}
     F(X,Y,Z) = \left( G_{r_1(X,Y,Z)} (X), G_{r_2(X,Y,Z)} (Y), G_{r_3(X,Y,Z)} (Z) \right)
 \end{equation*} 
 for some functions $r_1,r_2,r_3: [1,n_k]^d \to [0,n_k/n_{k-1}+1]$ we will define later.

\subsubsection*{Defining the target indices.}
We will now define the functions $r_1, r_2, r_3$ that specify the target indices of $X,Y$ and $Z$. First, for $X=(x_1, x_2, \dots, x_{d/3}) \in [1,n_k]^{d/3}$, let 
\begin{equation*}
    ind^{(k)}(X)= \max\{i:\forall r \in [1,d/3] \; x_r \geq (i-1)n_{k-1}+1 \} = 1+  \left\lfloor \frac{\min_r x_r -1}{n_{k-1}} \right\rfloor,
\end{equation*}
that is, $ind^{(k)}(X)$ is the index of the interval $I_i^{(k)}$ the smallest coordinate of $X$ is contained in. We also define the \emph{inner index} of $X$ relative to a hypercube $C_a$: let $\text{clip}_{C_a}(X)=X'=(x_1',x_2', \dots, x_{d/3}')$, and, for $k\geq2$, define

\begin{multline*}
\hspace{-5 mm}
     iind_a^{(k)}(X)=\max\{i:\forall r \in [1,d/3] \; x_r' \geq (a-1)n_{k-1}+(i-1)n_{k-2}+1 \}  \\
    = 1+\left\lfloor \frac{(\min_r x'_r)-(a-1)n_{k-1}-1}{n_{k-2}} \right\rfloor. 
\end{multline*}

For $k=1$, use the same definition with $m_0$ in place of $n_{k-2}$. 

We will also use the index of the group containing the inner interval $I^{(k-1)}_i$ with $i=iind_a^{(k)}(X)$. We group consecutive intervals according to the label map at level $k-1$, extending the same rule to all inner intervals, and define $iindg_a^{(k)}(X)= iind_a^{(1)}(X) \text{ if } k=1$, and
\begin{equation*}
    iindg_a^{(k)}(X)=
        1+\left\lfloor\dfrac{iind_a^{(k)}(X)}{8m_{k-1}}\right\rfloor
\end{equation*}
otherwise. We omit the superscript $(k)$ when the level is clear.
For $k=1$, each original interval forms its own group. The function $iindg_a^{(k)}$ is monotone, will equal $J(G_a)$ at the component fixed point $p_a$, and equals $m_{k-1}+1$ when $X$ clips to the upper corner of $C_a$.

For an illustration, see Figure \ref{fig:index}. 
\begin{figure}[htbp]
    \centering
\begingroup
\definecolor{tarskiink}{HTML}{243B4A}
\definecolor{tarskiblue}{HTML}{206793}
\definecolor{tarskiorange}{HTML}{AD5627}
\begin{tikzpicture}[x=1.05cm,y=1.05cm,>=Stealth,
    every node/.style={text=tarskiink},
    coarse/.style={draw=tarskiblue!40,line width=.55pt},
    fine/.style={draw=tarskiblue!45,line width=.45pt}]

    \foreach \i in {0,2,3}
        \fill[tarskiblue!7] (\i,\i) rectangle +(1,1);
    \fill[tarskiorange!8] (1,1) rectangle (2,2);
    \foreach \i in {0,...,4} {
        \draw[coarse] (\i,0)--(\i,4);
        \draw[coarse] (0,\i)--(4,\i);
    }
    \foreach \i in {1,2,3} {
        \draw[fine] (1+\i/4,1)--(1+\i/4,2);
        \draw[fine] (1,1+\i/4)--(2,1+\i/4);
    }
    \draw[tarskiorange,line width=1pt] (1,1) rectangle (2,2);
    \foreach \i in {1,3,4}
        \node at (\i-.5,\i-.5) {$C_\i$};
    \node[anchor=west] at (2.08,1.45) {$C_2$};
    \fill[tarskiblue] (20/12,35/12) circle (2.2pt);
    \node[anchor=west,inner sep=3pt] at (20/12,35/12) {$X$};
    \draw[->,tarskiorange,line width=.85pt]
        (20/12,35/12-.10)--(20/12,2+.10);
    \fill[tarskiorange] (20/12,2) circle (2.4pt);
    \node[anchor=south west,inner sep=3pt] at (20/12+.04,2+.03) {$X'$};
    \foreach \i in {1,...,4} {
        \node[below=4pt] at (\i-.5,0) {$I_\i$};
        \node[left=4pt] at (0,\i-.5) {$I_\i$};
    }
    \node[right=4pt] at (4,0) {$u$};
    \node[above=4pt] at (0,4) {$v$};
    \node at (2,4.55) {Outer blocks};
    \node at (2,-.95) {$\operatorname{ind}^{(1)}(X)=2$};

    \draw[->,tarskiorange!80,line width=.8pt] (4.4,2)--(5.35,2);

    \begin{scope}[xshift=6.0cm]
        \fill[tarskiorange!14] (2,3) rectangle (3,4);
        \foreach \i in {0,...,4} {
            \draw[fine] (\i,0)--(\i,4);
            \draw[fine] (0,\i)--(4,\i);
        }
        \draw[tarskiorange,line width=1pt] (0,0) rectangle (4,4);
        \draw[tarskiorange!70,densely dashed,line width=.6pt]
            (8/3,0)--(8/3,4);
        \fill[tarskiorange] (8/3,4) circle (2.8pt);
        \node[anchor=south west,inner sep=3pt] at (8/3+.04,4+.03) {$X'$};
        \foreach \i in {1,...,4} {
            \node[below=4pt] at (\i-.5,0) {$\i$};
            \node[left=4pt] at (0,\i-.5) {$\i$};
        }
        \node[right=4pt] at (4,0) {$r_u$};
        \node[above=4pt] at (0,4) {$r_v$};
        \node at (2,4.55) {Inside $C_2$};
        \node at (2,-.95)
            {$\operatorname{iind}^{(1)}_2(X)=\min\{3,4\}=3$};
    \end{scope}
\end{tikzpicture}
\endgroup
\caption{Outer and inner indices for one two-dimensional coordinate
    group of a six-dimensional input. The outer intervals have length
    $12$, and the inner intervals have length $3$.
    The point $X=(20,35)$ has outer index $2$.
    Clipping it to $C_2=[13,24]^2$ gives $X'=(20,24)$,
    whose local coordinates are $(8,12)$.
    Their inner interval indices are $3$ and $4$, respectively,
    so $\operatorname{iind}^{(1)}_2(X)=3$.
    The enlarged view labels these inner intervals by $r_u$ and $r_v$.}
\label{fig:index}
\end{figure}
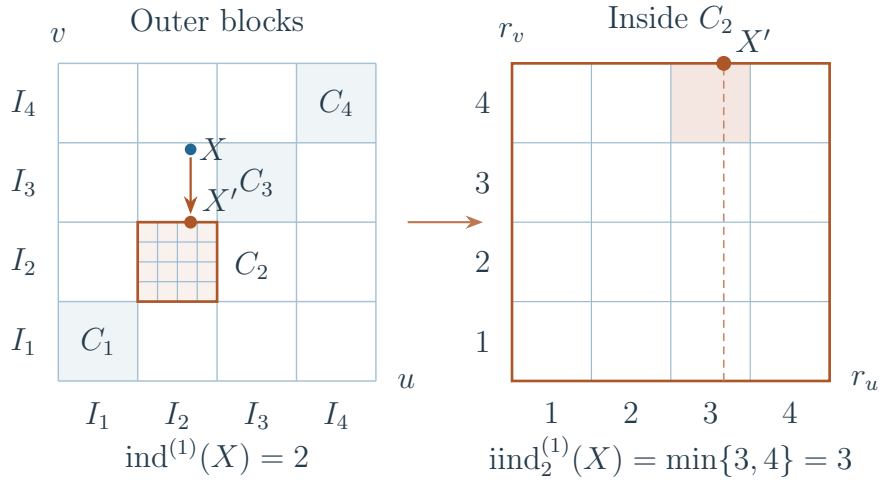

For an ordered pair of points $X,Y \in [1,n_k]^{d/3}$, set $cind(X,Y)= \min\{ind(X),ind(Y)-1\}$, and
\begin{equation*}
    \scalemath{0.93}{A(X,Y)=
    \begin{cases}
        0 & \text{ if } cind(X,Y) =0, \\
        cind(X,Y)-1+ 1_{\substack{iindg_{cind(X,Y)}(X)+iindg_{cind(X,Y)+1}(Y) \geq 4m_k^2+2m_k+3}} & \text{ otherwise.}
    \end{cases}}
\end{equation*}

Let $S=A(X,Y)+A(X,Z)+A(Y,Z)$. We are now ready to define the index functions $r_1,r_2,r_3$. We distinguish three cases based on the remainder of $S$ modulo 3:

\begin{equation*}
\left(r_1(X,Y,Z), r_2(X,Y,Z), r_3(X,Y,Z) \right) = 
\begin{cases}
    \left(i,  i+1, i+2\right)  & \text{if } S=3i,\\
    \left(i+1,  i+1, i+2\right)  & \text{if } S=3i+1,\\
    \left(i+1,  i+2, i+2\right)  & \text{if } S=3i+2.
\end{cases}
\end{equation*}

Substituting in the definition of $F$, we get 
\begin{equation}\label{eq:Fdef}  
F(X,Y,Z) = 
\begin{cases}
    \left(G_i(X),  G_{i+1}(Y), G_{i+2}(Z)\right)  & \text{if } S=3i,\\
    \left(G_{i+1}(X),  G_{i+1}(Y), G_{i+2}(Z)\right)  & \text{if } S=3i+1,\\
    \left(G_{i+1}(X),  G_{i+2}(Y), G_{i+2}(Z)\right)  & \text{if } S=3i+2.
\end{cases}
\end{equation}

For an illustration, see Figure \ref{fig:staircase}.

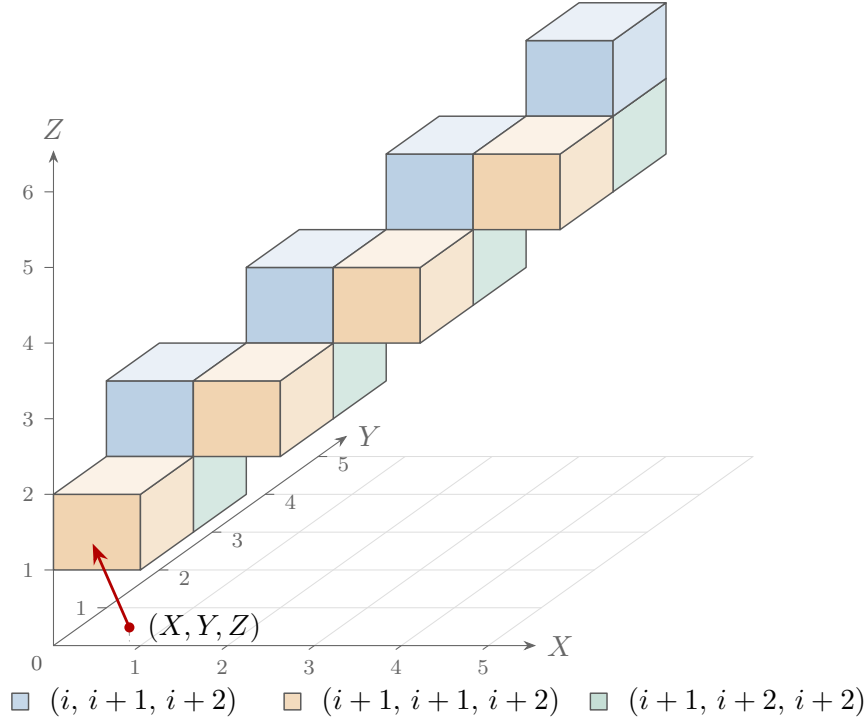
\begin{figure}[tbp]
  \centering
  \begin{tikzpicture}[
      x={(1.15cm,0cm)}, y={(.70cm,.50cm)}, z={(0cm,1.00cm)},
      line join=round, line cap=round,
      staircase edge/.style={draw=black!65,line width=.55pt},
      >=Stealth]
    \definecolor{staircaseBlue}{RGB}{92,143,191}
    \definecolor{staircaseOrange}{RGB}{221,153,69}
    \definecolor{staircaseGreen}{RGB}{91,164,139}
    \newcommand{\staircasecube}[4]{%
      \begin{scope}[shift={(#1,#2,#3)}]
        \path[staircase edge,fill=#4!42]
          (0,0,0)--(1,0,0)--(1,0,1)--(0,0,1)--cycle;
        \path[staircase edge,fill=#4!25]
          (1,0,0)--(1,1,0)--(1,1,1)--(1,0,1)--cycle;
        \path[staircase edge,fill=#4!13]
          (0,0,1)--(1,0,1)--(1,1,1)--(0,1,1)--cycle;
      \end{scope}}

    \foreach \u in {0,...,5} {
      \draw[black!13,line width=.35pt] (\u,0,0)--(\u,5,0);
      \draw[black!13,line width=.35pt] (0,\u,0)--(5,\u,0);
    }
    \draw[->,black!60] (0,0,0)--(5.55,0,0)
      node[right,font=\small] {$X$};
    \draw[->,black!60] (0,0,0)--(0,5.55,0)
      node[right,font=\small] {$Y$};
    \draw[->,black!60] (0,0,0)--(0,0,6.55)
      node[above,font=\small] {$Z$};
    \node[below left,font=\scriptsize,black!65] at (0,0,0) {$0$};
    \foreach \t in {1,...,5} {
      \draw[black!60] (\t,0,0)--++(0,-.10,0)
        node[below,font=\scriptsize] {$\t$};
      \ifnum\t=1
        \draw[black!60] (0,\t,0)--++(-.10,0,0)
          node[left,font=\scriptsize] {$\t$};
      \else
        \draw[black!60] (0,\t,0)--++(.10,0,0)
          node[right,yshift=-3pt,font=\scriptsize] {$\t$};
      \fi
    }
    \foreach \t in {1,...,6} {
      \draw[black!60] (0,0,\t)--++(-.10,0,0)
        node[left,font=\scriptsize] {$\t$};
    }

    \staircasecube{0}{1}{1}{staircaseGreen}
    \staircasecube{1}{2}{2}{staircaseGreen}
    \staircasecube{2}{3}{3}{staircaseGreen}
    \staircasecube{3}{4}{4}{staircaseGreen}
    \staircasecube{0}{1}{2}{staircaseBlue}
    \staircasecube{1}{2}{3}{staircaseBlue}
    \staircasecube{2}{3}{4}{staircaseBlue}
    \staircasecube{3}{4}{5}{staircaseBlue}
    \staircasecube{0}{0}{1}{staircaseOrange}
    \staircasecube{1}{1}{2}{staircaseOrange}
    \staircasecube{2}{2}{3}{staircaseOrange}
    \staircasecube{3}{3}{4}{staircaseOrange}

    \coordinate (sampleInput) at (.80,.12,.18);
    \draw[black!35,densely dashed,line width=.45pt]
      (sampleInput)--(.80,.12,0);
    \draw[->,red!70!black,line width=1.05pt]
      (sampleInput)--(.45,0,1.35);
    \fill[red!70!black] (sampleInput) circle (1.9pt);
    \node[anchor=west,font=\small,inner sep=2pt]
      at (.94,.12,.18) {$(X,Y,Z)$};

    \begin{scope}[x={(1cm,0cm)},y={(0cm,1cm)},font=\small]
      \path[staircase edge,fill=staircaseBlue!35]
        (-.55,-.85) rectangle ++(.22,.22);
      \node[anchor=west] at (-.20,-.74) {$(i,\,i+1,\,i+2)$};
      \path[staircase edge,fill=staircaseOrange!35]
        (3.05,-.85) rectangle ++(.22,.22);
      \node[anchor=west] at (3.40,-.74) {$(i+1,\,i+1,\,i+2)$};
      \path[staircase edge,fill=staircaseGreen!35]
        (7.10,-.85) rectangle ++(.22,.22);
      \node[anchor=west] at (7.45,-.74) {$(i+1,\,i+2,\,i+2)$};
    \end{scope}
  \end{tikzpicture}
    \caption{The three axes represent the coordinate blocks $X$, $Y$, and $Z$, each of dimension $d/3$.
    The staircase begins at $C_1\times C_1\times C_2$, selected for both $S=0$ and $S=1$.
    Integer ticks mark coarse interval boundaries.
    The marked input has all three indices equal to $1$; the arrow points to its target box.}
  \label{fig:staircase}
\end{figure}

We are now ready to define the function class $\FF_k$. Fix a switching index $j \in [1,8m_k^2-1]$, we will describe the elements of $\FF_k$ with switching index $j$.

Choose any sequence of functions $G_1', G_2', \dots, G_{n_k/n_{k-1}}' \in \FF_{k-1}$ such that their labels satisfy the condition
\begin{equation}\label{eq:Fcond2}
    J(G_r') + J(G_{r+1}') \geq 4m_k^2 +2m_{k}+3  \Longleftrightarrow r \leq j \qquad (1\leq r<n_k/n_{k-1}),
\end{equation}
and modify them as defined in (\ref{eq:Gdef}) to get functions $G_0, G_1, \dots, G_{n_k/n_{k-1}+1}$.

The family $\FF_k$ consists of the functions $F:[1,n_k]^d\to[1,n_k]^d$ obtained by (\ref{eq:Gdef}) and (\ref{eq:Fdef}) from functions $G_1',\ldots,G_{n_k/n_{k-1}}'\in\FF_{k-1}$ satisfying (\ref{eq:Fcond2}) with threshold $4m_k^2 + 2m_k +3$ and a switching index $j\in[1,8m_k^2-1]$.
We define the label of $F$ by
\begin{equation*}
    J(F)=1+\left\lfloor\frac{j}{8m_k}\right\rfloor.
\end{equation*}

For $i\in [1,8m_k^2-1]$, the label group containing $C_i$ is given by the same
expression with $i$ in place of $j$.

\begin{lemma}\label{lem:Fmonuniqgen}
    For any choice of $1 \leq j < 8m_k^2$, and  $G_1', \dots, G_{n_k/n_{k-1}}'$ satisfying condition (\ref{eq:Fcond2}), the function $F$ defined by equation (\ref{eq:Fdef}) using (\ref{eq:Gdef}) is a monotone function with a unique fixed point $(p_j,p_{j+1},p_{j+2}) \in C_j \times C_{j+1} \times C_{j+2}$, where $p_r$ is the unique fixed point of $G_r$ for each $r$.
\end{lemma}
\begin{proof}
    We use induction on $k$. The case $k=0$ was covered by Lemma \ref{lem:fmonuniq0}. 
    By the induction hypothesis, the functions $G_r'$, and therefore also $G_r$, are monotone functions with a unique fixed point.

    For monotonicity, first see that the function $A$ is monotone: for fixed $cind(X,Y)$ this follows from the monotonicity of $iindg_r$ for a fixed index $r$, and if $cind(X,Y)$ increases, the possible values of $A$ go from $\{cind(X,Y)-1, cind(X,Y)\}$ to a set whose minimum is $cind(X,Y)$.

    This implies the monotonicity of $S$, and therefore, also of $r_1, r_2$, and $r_3$.

    Finally, the property 
    \begin{equation*}
            G_{i'} (X') \geq  G_i (X) \text{ for all } i' \geq i , \text{ and }X'\geq X
    \end{equation*}
    together with the monotonicity of each $G_i$, $r_1, r_2,$ and $ r_3$ proves the monotonicity of $F$.
    
    Fix a function $F$ and suppose that $(X,Y,Z)$ is a fixed point of $F$.

    For two component fixed points $p_s$ and $p_t$ for $1 \leq s,t \leq n_k/n_{k-1}$, 
    \begin{itemize}
        \item $t \leq s$ implies $cind(p_s, p_t) = t-1$. If $t=1$, then $A(p_s,p_t)=0$, otherwise $iindg_{cind(p_s,p_t)}(p_s) = m_{k-1}+1$,
        \item $t = s+1$ implies $cind(p_s, p_t) = t-1=s$ and thus $iindg_{cind(p_s,p_t)}(p_s)+iindg_{cind(p_s,p_t)+1}(p_t) =J(G_s)+J(G_t)$,
        \item $t \geq s+2$ implies $cind(p_s, p_t) = s$ and thus $iindg_{cind(p_s,p_t)+1}(p_t) = m_{k-1}+1$.
    \end{itemize}
    Therefore, by condition (\ref{eq:Fcond2}),
    \begin{equation*}
        A(p_s,p_t) = \begin{cases}
            t-1 & \text{ if } t \leq s, \\
            s-1 +1_{s \leq j} & \text{ if } t = s+1, \\
            s  & \text{ if } t \geq s+2.
        \end{cases}
    \end{equation*}
Since $0\le A\le n_k/n_{k-1}-1$, we have $0\le S\le 3n_k/n_{k-1}-3$.
We first exclude the two boundary cases.

If $S=0$, the fixed-point equations force
$(X,Y,Z)=(p_1,p_1,p_2)$. Since $j\ge1$, this gives $S=A(p_1,p_1)+2A(p_1,p_2)=0+2=2$, a contradiction.

If $S=3n_k/n_{k-1}-3$, they force
$(X,Y,Z)=(p_{n_k/n_{k-1}-1},p_{n_k/n_{k-1}},p_{n_k/n_{k-1}})$. Since $j\le n_k/n_{k-1}-2$,
this gives $S=2A(p_{n_k/n_{k-1}-1},p_{n_k/n_{k-1}})+A(p_{n_k/n_{k-1}},p_{n_k/n_{k-1}})
  =2(n_k/n_{k-1}-2)+(n_k/n_{k-1}-1)=3n_k/n_{k-1}-5$,
again a contradiction.

Now $1 \leq S \leq 3n_k/n_{k-1}-4$, so all selected indices lie in $[1,n_k/n_{k-1}]$. 
For the rest of the proof, we let $r_1=r_1(X,Y,Z), r_2=r_2(X,Y,Z), $ and $r_3=r_3(X,Y,Z)$.
As the image of $G_i$ is contained in $C_i$, this implies 
$X \in C_{r_1}, Y \in C_{r_2},$ and $Z \in C_{r_3}$, and furthermore, that
$X= p_{r_1} $, $Y= p_{r_2}$, $Z= p_{r_3}$.

Again, we distinguish the three cases of the remainder of $S$ modulo 3. 
\begin{equation*} 
\begin{array}{c|c|c}
S & (X,Y,Z) &
A(X,Y)+A(X,Z)+A(Y,Z)\\ \hline
3i
& (p_i,p_{i+1},p_{i+2})
& 3i-1 +1_{i \leq j} + 1_{i+1 \leq j}\\
3i+1
& (p_{i+1},p_{i+1},p_{i+2})
& 3i+2 \cdot 1_{i+1 \leq j}\\
3i+2
& (p_{i+1},p_{i+2},p_{i+2})
& 3i+1+2\cdot 1_{i+1 \leq j}
\end{array}
\end{equation*}

In the second and third rows, equality between the first and last columns would require $2\cdot 1_{i+1 \leq j}=1$, which is impossible. In the first row, equality holds precisely when $1_{i \leq j}+ 1_{i+1 \leq j}=1$, which is equivalent to $i=j$. Consequently, the only possible fixed point is $(p_j,p_{j+1},p_{j+2})$.

Conversely, evaluating the defining sum at this triple gives $S=3j$, so its selected indices are $(j,j+1,j+2)$.
Since each $p_r$ is fixed by $G_r$, this triple is indeed a fixed point of $F$.

\end{proof}

\subsection{The herringbone family}

For the answering strategy of $\Sigma_k$ to queries of the form $J(G_i) \overset{?}{\leq} t$, we need to define a family of two-dimensional Tarski instances called \emph{herringbones}.
This class was introduced by Etessami et al. in \cite{EPRY20} to show that $\text{Tarski}(n,2) = \Omega(\log^2 n)$.
Here, we will need a slightly different hardness property:  a simulator that answers evaluation and sublabel comparison queries using few comparisons about the label of $H$.

First, we define the family of paths that can be the spines of our family of herringbone functions.
Fix a label $j \in [1,m_k]$, we define paths that belong to herringbone functions with label $j$.

Each path is defined using $m_k$ parameters $c_1, c_2, \dots, c_{m_k} \in [-m_k,m_k]$:

Divide the grid into $2m_k+1$ \emph{stripes} defined by 
\begin{equation*}
    L_i=\{ (x,y) \in [1,4m_k^2+2m_k]^2  : 4(i-1)m_k+1 \leq x+y \leq 4im_k\}.
\end{equation*} 
In stripes with even index, the path moves on a staircase located on the diagonals $c_i$ and $c_{i}+1$: $P \cap L_{2i} = \{(x,y) \in L_{2i}: x-y \in \{c_{i},c_{i}+1 \}\}$.
Stripes with odd index are called \emph{transition stripes}.
Let $A_i = (x_{A_i}, y_{A_i})$ and $B_i = (x_{B_i}, y_{B_i})$ be the last spine point of $L_{2i}$ and the first spine point of $L_{2i+2}$, prescribed by the adjacent staircases with parameters $c_i$ and $c_{i+1}$, respectively.
The stripe width ensures that $A_i\leq B_i$ coordinatewise.

Within the transition stripe, the path first moves horizontally from $A_i$ to $(x_{B_i},y_{A_i})$, and then vertically to $B_i$, always taking unit steps.
For the first and last transition stripes, use $(1,1)$ and $(4m_k^2+2m_k, 4m_k^2+2m_k)$ as the outer endpoints, respectively.

For an example, see Figure \ref{fig:herringbone}.

\input{figures/herringbone_overview.tikz}

Write $P=(p_2,\ldots,p_{8m_k^2+4m_k})$, where $p_2=(1,1)$, $p_{8m_k^2+4m_k}=(4m_k^2+2m_k,4m_k^2+2m_k)$, and $p_t=(x_t,y_t)$ satisfies $x_t+y_t=t$.

The functions $H:[4m_k^2+2m_k]^2\to[4m_k^2+2m_k]^2$ with spine $P$ and label $j$ are defined as follows:
\begin{itemize}
    \item Choose an index $t^*$ such that
    $p_{t^*}$ lies in an even-indexed stripe and $ y_{t^*} \in [4(j-1)m_k+1 , 4jm_k]$.

    \item On the spine, all points move towards $p_{t^*}$:
    \[
        H(p_t)=
        \begin{cases}
            p_{t+1}, & t<t^*,\\
            p_t,     & t=t^*,\\
            p_{t-1}, & t>t^*.
        \end{cases}
    \]

    \item Outside the spine, each point moves towards the spine while preserving the sum of its coordinates.
    For $(x,y)\notin P$, writing $t=x+y$, set
    \[
        H(x,y)=
        \begin{cases}
            (x-1,y+1), & x>x_t,\\
            (x+1,y-1), & x<x_t.
        \end{cases}
    \]
\end{itemize}
Thus $H$ has the unique fixed point $p_{t^*}=(x^*,y^*)$.
Let $\HH_k$ be the family of all functions obtained in this way. We set $J(H)=j=\lceil y^*/(4m_k)\rceil$ and define its sublabels $\lambda_i=\lambda_i(H)$ for $i\in[1,8m_k^2+4m_k+1]$ as follows:

For each $y\in[4m_k^2+2m_k]$, let $a_y=\min\{x:(x,y)\in P\}, b_y=\max\{x:(x,y)\in P\} $
be the first and last coordinates visited by the path in row $y$.
Set $\lambda_1=4m_k^2+2m_k+1,\lambda_{8m_k^2+4m_k+1}=0.$
For each $y\in[1, 4m_k^2+2m_k]$, define
\[
    \lambda_{2y}=
    \begin{cases}
        b_y,   & \text{ if } y<y^*,\\
        x^*-1, & \text{ if } y=y^*,\\
        a_y-1, & \text{ if }y>y^*,
    \end{cases}
\]
and, for $2\leq y\leq 4m_k^2+2m_k$, define
\[
    \lambda_{2y-1}=4m_k^2+2m_k+1-a_y.
\]
Thus every interior sublabel belongs to $\{0,\ldots,4m_k^2+2m_k\}$.

\begin{lemma}\label{lem:sublabels}
    For all choices of $j \in [1,m_k]$ and $c_i \in [-m_k, m_k]$, the resulting function is monotone with a unique fixed point $p_{t^*}=(x^*,y^*)$ and sublabels satisfying 
    \begin{equation*}
        \lambda_i+ \lambda_{i+1} \geq 4m_k^2+2m_k+1 \Longleftrightarrow  i \leq 2y^*-2 +1_{x^* > a_{y^*}} \qquad (i\in [1,8m_k^2+4m_k]).
    \end{equation*}
\end{lemma}
\begin{proof}
To prove monotonicity, consider two comparable adjacent grid points. If neither lies on $P$, they are on the same side of $P$ and receive the same translation. If both lie on $P$, their images remain ordered, since all spine points move towards $p_{t^*}$. 
Otherwise, the off-spine point maps to the predecessor or successor of the spine point, according as it is the smaller or larger point. The spine point cannot move past that neighbor, so order is again preserved. 
Moreover, every point except $p_{t^*}$ moves, proving uniqueness of the fixed point.

Since $P$ moves by unit steps, we have $a_{y+1}=b_y$.
Substituting the definitions of the sublabels gives
\begin{equation*}
    \begin{aligned}
 &\lambda_{2y}+\lambda_{2y+1}\geq4m_k^2+2m_k+1
  \quad\Longleftrightarrow\quad y<y^*,\\
&\lambda_{2y-1}+\lambda_{2y}\geq 4m_k^2+2m_k+1
\quad\Longleftrightarrow\quad
  y<y^*\ \text{or}\ (y=y^*\text{ and }x^*>a_{y^*}).
\end{aligned}
\end{equation*}

Thus the inequality holds precisely up to index
$2y^*-2+1_{x^*>a_{y^*}}$, as claimed.
\end{proof}

\begin{lemma}\label{lem:herringbone-simulation} 
There exists an online simulator $\Sigma_k^{\mathrm H}$
with the following property.
Given access to an unknown admissible label $j$ only through comparison queries of the form $j\overset{?}{\leq}t$, the simulator answers any adaptively chosen sequence of evaluation queries $H(x,y)$ and sublabel comparison queries $\lambda_i\overset{?}{\leq}t$, where $(x,y)\in[4m_k^2+2m_k]^2$, while maintaining the following properties:
\begin{itemize}
    \item At every stage, for every admissible label $j'$
    consistent with all answers to the comparison queries
    made by $\Sigma_k^{\mathrm H}$ so far, there exists
    a function $H'\in\HH_k$ with $J(H')=j'$ that is
    consistent with all evaluation and sublabel comparison answers returned by
    $\Sigma_k^{\mathrm H}$ so far.

    \item After receiving $Q$  queries, the
    simulator has made at most $O(Q/\log m_k)$ comparison
    queries about $j$.
\end{itemize}
\end{lemma}
\begin{proof}
First, note that each sublabel comparison query can be answered
using at most one evaluation query: for $t\in\{0,\ldots,4m_k^2+2m_k-1\}$,
\[
\begin{aligned}
    \lambda_{2y}\leq t
    &\Longleftrightarrow
    H(t+1,y)\in\{(t,y+1),(t,y),(t+1,y-1),(t+1,y)\}
    &&(y\in[1, 4m_k^2+2m_k]),\\
    \lambda_{2y-1}\leq t
    &\Longleftrightarrow
    H(4m_k^2+2m_k-t,y)=(4m_k^2+2m_k+1-t,y-1)
    && (y \in [2, 4m_k^2+2m_k]).
\end{aligned}
\]
Therefore, in the rest of the proof, we only discuss evaluation queries.
We maintain the additional requirement that the fixed point lies in $L_{2j}$, with $4(j-1)m_k<y^*\leq4jm_k$.

For each $i$, the simulator maintains an interval $C_i$ of admissible values of $c_i$. Initially, $C_i = [-m_k, m_k]$ for each $i$. On query $(x,y)$, consider first the case when the queried point $(x,y)$ is in a stripe with even index $L_{2i}$. Let $C_i=[a,b]$ for some $a\leq b$.
If $\max\{ x-y -a, b-(x-y)+1 \} \geq 2$, the query is not forced to lie on $P$.
In this case, the simulator answers $H(x,y) =(x-1, y+1)$  and updates $C_i= \left[ a,\min\{b, x-y-2 \}\right] $ if $ x-y -a > b-(x-y)$, while it answers $H(x,y) =(x+1, y-1)$  and updates $C_i= \left[\max\{a, x-y+1\}  ,b \right]$ if $ x-y -a \leq b-(x-y)$. 
If $\max\{ x-y -a, b-(x-y)+1 \} < 2$, the queried point is forced to be on the spine $P$ and either $C_i=\{x-y-1\}$, $C_i=\{x-y\}$, or $C_i=\{x-y-1,x-y\}$ is also forced. In the latter case, fix $c_i=x-y$. 
The simulator now determines  $j \overset{?}{\leq} i$ and $j \overset{?}{\leq} i-1$ using at most two queries (reusing previous answers if possible), and concludes $j > i$, $j<i$, or $j=i$. 
If $j=i$ and the fixed point has not yet been chosen, set $p_{t^*}=(x,y)$ if $y\leq4im_k$; otherwise, choose the first spine point of $L_{2i}$ as $p_{t^*}$. The simulator answers according to the definition of $H$: it returns $(x,y)$ if this is the chosen fixed point, and otherwise the neighboring point of $P$ in the direction of the fixed point. This direction is determined by $j$ or by the already chosen fixed point.
    
If the queried point is in an interior stripe with odd index $L_{2i+1}$, the answer depends on both $c_i$ and $c_{i+1}$.

Consider a query $(x,y)\in L_{2i+1}$, where $i \in [1, m_k-1]$.
In this stripe, the spine first moves horizontally at height $y_{A_i}$ and then vertically at first coordinate $x_{B_i}$, and the answer to a query depends only on $A_i$, $B_i$, and $j$, as summarized in the following table.

\[
\begin{array}{c|c|c}
    \text{Condition}
    & H(x,y)\text{ if }j\leq i
    & H(x,y)\text{ if }j>i\\
    \hline
    x>x_{B_i}\text{ or }y<y_{A_i}
    & (x-1,y+1) & (x-1,y+1)\\
    x<x_{B_i}\text{ and }y>y_{A_i}
    & (x+1,y-1) & (x+1,y-1)\\
    x<x_{B_i}\text{ and }y=y_{A_i}
    & (x-1,y) & (x+1,y)\\
    x=x_{B_i}\text{ and }y>y_{A_i}
    & (x,y-1) & (x,y+1)\\
    x=x_{B_i}\text{ and }y=y_{A_i}
    & (x-1,y) & (x,y+1)
\end{array}
\]

\smallskip

Moreover, the comparisons between $y$ and $y_{A_i}$, and between $x$ and $x_{B_i}$, depend only on $c_i$ and $c_{i+1}$, respectively.
The simulator answers them independently using the interval restriction from the even-stripe case, with $x-y$ replaced by $8im_k-2y$ and $2x-4(2i+1)m_k-1$, respectively.

The answer is then given by the table.
If it requires knowing whether $j\leq i$, then $(x,y)\in P$, so at least one of the two comparisons is an equality.
The corresponding parameter has therefore been fixed, and the simulator determines $j\overset{?}{\leq}i$ using at most one query, reusing a previous answer whenever possible.
In the first and last transition stripes, only one adjacent parameter is involved. The
same restriction rule applies; on-spine points move forward in the first stripe and backward
in the last, so no label comparison is needed there.

\begin{figure}[htbp]
\centering
\begingroup
\definecolor{hbSpine}{HTML}{185A78}
\definecolor{hbStripe}{HTML}{EDF3F6}
\definecolor{hbGrid}{HTML}{DDE1E4}
\definecolor{hbField}{HTML}{9BA5AC}
\definecolor{hbInk}{HTML}{27343D}
\definecolor{hbAccent}{HTML}{B34D2F}

\begin{tikzpicture}[
  x=0.50cm,y=0.50cm,
  font=\small, text=hbInk,
  hbfield/.style={draw=hbField,line width=0.4pt,-{Stealth[length=1.6mm,width=1.05mm]}},
  hbpatharrow/.style={draw=hbSpine,line width=1.1pt,-{Stealth[length=1.9mm,width=1.35mm]}},
  stripe label/.style={fill=white,fill opacity=0.94,text opacity=1,inner sep=2.5pt},
]
\begin{scope}
  \clip (7.5,9.5) rectangle (22.5,21.5);
  \fill[hbStripe] (7.5,9.5) -- (15,9.5) -- (7.5,17) -- cycle;
  \fill[hbStripe] (22.5,14) -- (22.5,21.5) -- (15,21.5) -- cycle;
  \draw[step=1,draw=hbGrid,line width=0.16pt] (7,9) grid (23,22);
  \foreach \s in {24.5,36.5} {
    \draw[draw=hbField,densely dashed,line width=0.5pt]
      ({\s-22},22) -- ({\s-9},9);
  }

  \draw[hbfield] (8,18) -- (9,17);
  \draw[hbfield] (8,20) -- (9,19);
  \draw[hbfield] (10,16) -- (11,15);
  \draw[hbfield] (10,18) -- (11,17);
  \draw[hbfield] (10,20) -- (11,19);
  \draw[hbfield] (12,14) -- (13,13);
  \draw[hbfield] (12,16) -- (13,15);
  \draw[hbfield] (12,18) -- (13,17);
  \draw[hbfield] (12,20) -- (13,19);
  \draw[hbfield] (14,12) -- (13,13);
  \draw[hbfield] (14,18) -- (15,17);
  \draw[hbfield] (14,20) -- (15,19);
  \draw[hbfield] (16,10) -- (15,11);
  \draw[hbfield] (16,12) -- (15,13);
  \draw[hbfield] (16,18) -- (17,17);
  \draw[hbfield] (16,20) -- (17,19);
  \draw[hbfield] (18,10) -- (17,11);
  \draw[hbfield] (18,12) -- (17,13);
  \draw[hbfield] (18,14) -- (19,13);
  \draw[hbfield] (18,16) -- (19,15);
  \draw[hbfield] (18,18) -- (19,17);
  \draw[hbfield] (20,10) -- (19,11);
  \draw[hbfield] (20,12) -- (19,13);
  \draw[hbfield] (20,14) -- (19,15);
  \draw[hbfield] (20,16) -- (19,17);
  \draw[hbfield] (22,10) -- (21,11);
  \draw[hbfield] (22,12) -- (21,13);
  \draw[hbfield] (22,14) -- (21,15);

  \draw[hbInk,dashed,line width=0.5pt] (7.5,13) -- (22.5,13);
  \draw[hbInk,dashed,line width=0.5pt] (19,9.5) -- (19,21.5);

  \draw[hbSpine,line width=1.55pt,line join=round,line cap=round]
    (7,9) -- (8,9) -- (8,10) -- (9,10) -- (9,11) -- (10,11) -- (10,12) -- (11,12) -- (11,13) -- (12,13) -- (13,13) -- (14,13) -- (15,13) -- (16,13) -- (17,13) -- (18,13) -- (19,13) -- (19,14) -- (19,15) -- (19,16) -- (19,17) -- (19,18) -- (19,19) -- (20,19) -- (20,20) -- (21,20) -- (21,21) -- (22,21) -- (22,22) -- (23,22) -- (23,23);
  \draw[hbpatharrow] (9,11) -- (10,11);
  \draw[hbpatharrow] (11,12) -- (11,13);
  \draw[hbpatharrow] (12,13) -- (13,13);
  \draw[hbpatharrow] (15,13) -- (16,13);
  \draw[hbpatharrow] (18,13) -- (19,13);
  \draw[hbpatharrow] (19,14) -- (19,15);
  \draw[hbpatharrow] (19,16) -- (19,17);
  \draw[hbpatharrow] (19,18) -- (19,19);
  \draw[hbpatharrow] (20,19) -- (20,20);
  \draw[hbpatharrow] (21,21) -- (22,21);
\end{scope}

  \draw[hbInk,line width=0.45pt] (7.5,9.5) rectangle (22.5,21.5);
  \node[stripe label,text=hbSpine] at (9.1,10.6) {$L_{2i}$};
  \node[stripe label] at (10.4,20.4) {$L_{2i+1}$};
  \node[stripe label,text=hbSpine] at (21,20.4) {$L_{2i+2}$};

  \node[anchor=east] at (7.4,13) {$y_{A_i}$};
  \node[anchor=south] at (19,21.7) {$x_{B_i}$};

  \fill[hbSpine] (11,13) circle[radius=1.8pt];
  \fill[hbSpine] (19,18) circle[radius=1.8pt];
  \fill[hbSpine] (19,13) circle[radius=1.8pt];
  \node[anchor=north east,fill=white,inner sep=1.8pt]
    at (10.85,12.78) {$A_i$};
  \node[anchor=south east,fill=white,inner sep=1.8pt]
    at (18.75,18.25) {$B_i$};
  \node[anchor=north west,fill=white,inner sep=2pt]
    at (19.2,12.7) {$(x_{B_i},y_{A_i})$};

  \draw[hbAccent,densely dotted,line width=0.65pt]
    (14,13) -- (14,16) -- (19,16);
  \node[anchor=east,fill=white,inner sep=1.8pt,text=hbAccent]
    at (13.78,14.6) {$y>y_{A_i}$};
  \node[anchor=south,fill=white,inner sep=1.8pt,text=hbAccent]
    at (16.5,16.16) {$x<x_{B_i}$};
  \draw[hbAccent,line width=1.1pt,-{Stealth[length=2.1mm,width=1.4mm]}]
    (14,16) -- (15,15);
  \fill[hbAccent] (14,16) circle[radius=2pt];
  \node[anchor=south east,fill=white,inner sep=1.8pt,text=hbAccent]
    at (13.85,16.15) {$(x,y)$};
  \node[anchor=north west,fill=white,inner sep=1.8pt,text=hbAccent]
    at (15.15,14.95) {$H(x,y)$};

\end{tikzpicture}
\endgroup
\caption{A query in the transition stripe $L_{2i+1}$. Comparing $y$ with
$y_{A_i}$ and $x$ with $x_{B_i}$ gives comparisons involving only $c_i$
and $c_{i+1}$, respectively. Together they determine the off-spine answer;
on the spine, the direction also depends on whether $j\leq i$.}
\label{fig:herringbone_transition}
\end{figure}
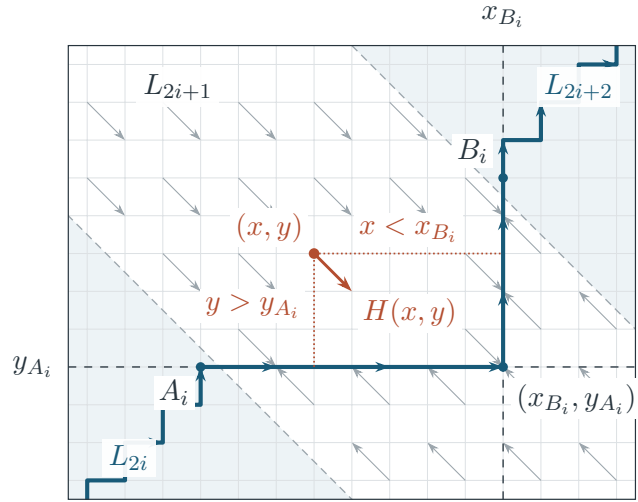

For consistency, take any label $j$ consistent with the comparison answers, and choose an arbitrary remaining value of each $c_i$ from $C_i$.
If no fixed point has been chosen, choose one in the corresponding even stripe $L_{2j}$ with $4(j-1)m_k < y^* \leq 4jm_k$. Otherwise, retain the chosen fixed point; in this case the label has already been determined.

The resulting function reproduces all off-spine answers by the preserved geometric comparisons.
The on-spine answers also agree, since their directions were determined by comparisons about $j$ or by the already chosen fixed point.
Thus every remaining label admits a function consistent with all answers given so far.

Finally, when making a new query $j \overset{?}{\leq} i$, at least one of $c_i$ and $c_{i+1}$ has been fixed, which requires $|C_i| \leq 3$ or $|C_{i+1}| \leq 3$, respectively. 
Each parameter $c_i$ can cause at most two new label queries, namely $j \overset{?}{\leq} i-1$ and $j \overset{?}{\leq} i$.
As each query influences at most two intervals $C_i$, and retains at least $|C_i|/2-1$ elements from said intervals, the number of label queries made by $\Sigma_k^H$ is at most $O(Q/\log m_k)$ after receiving $Q$ queries.

\end{proof}

\subsection{Function of the simulator $\Sigma_k$}

Finally, we will describe the operation of the simulator $\Sigma_k$ such that it satisfies the consistency and the query bound conditions.

First, notice that each evaluation query $F(X,Y,Z)$ can be answered using at most three evaluation queries to the component functions $G_i$. Indeed, the values $A(X,Y)$, $A(X,Z)$ and $A(Y,Z)$ depend only on the input coordinates, so we can compute $S$ and the target indices $r_1(X,Y,Z)$, $r_2(X,Y,Z)$ and $r_3(X,Y,Z)$ without querying any component function.
We then make the queries
    $G_{r_1(X,Y,Z)}(X),
    G_{r_2(X,Y,Z)}(Y),
    G_{r_3(X,Y,Z)}(Z)$,
and concatenate the three returned vectors to obtain $F(X,Y,Z)$.

By the induction hypothesis, the simulator $\Sigma_{k-1}$ exists and satisfies all the properties stated in the preliminaries at level $k-1$. For each component $G_i'$, we use a separate copy of this simulator, denoted by $\Sigma_{k-1,i}$. We also refer to these as \emph{subsimulators}.

To answer a component query $G_i(U)$, we clip and translate $U$ as in (\ref{eq:Gdef}), forward the resulting evaluation query to $\Sigma_{k-1,i}$, and translate its answer back to $C_i$. 
The endpoint copies $G_0$ and $G_{n_k/n_{k-1}+1}$ use the same subsimulators as $G_1$ and $G_{n_k/n_{k-1}}$, respectively. 

As a reminder, $\Sigma_k$ has access to $F$ only through queries of the form $J(F) \overset{?}{\leq} t$, and it must answer queries of $\AA_k$ such that for every index $J$ that is still consistent with the queries of $\Sigma_k$ there must exist an $F \in \FF_k$ with $J(F) =J$ that is consistent with all answers given to $\AA_k$ so far.

Furthermore, after $\AA_k$ has made $Q$ queries, the number of queries made by $\Sigma_k$ is bounded from above by $O_k\left( \frac{Q}{\log^{k+1}n_k}\right)$.

To achieve this, the simulator $\Sigma_k$ first translates a query made by $\AA_k$ to at most three queries about the functions $G_i$, and forwards each such query to the corresponding $\Sigma_{k-1,i}$. The subsimulator either returns an answer immediately, or its answer depends on one or more queries of the form $J(G_i) \overset{?}{\leq } t$.

When $\Sigma_k$ receives a query $J(G_i) \overset{?}{\leq } t$ from a subsimulator $\Sigma_{k-1,i}$, it makes the sublabel query $\lambda_i \overset{?}{\leq } t-1$ to $\Sigma_k^H$, and returns its answer to $\Sigma_{k-1,i}$. 
While answering such a sublabel query, the simulator $\Sigma_k^H$ might first ask queries $J(H) \overset{?}{\leq} t$, each of which $\Sigma_k$ forwards as $J(F) \overset{?}{\leq} t$ to the comparison oracle $\mathcal{O}_k$; $\Sigma_k$ then returns each answer to $\Sigma_k^H$.
For a schematic map of the operation of $\Sigma_k$, see Figure \ref{fig:query_interfaces_five_enclosed}.

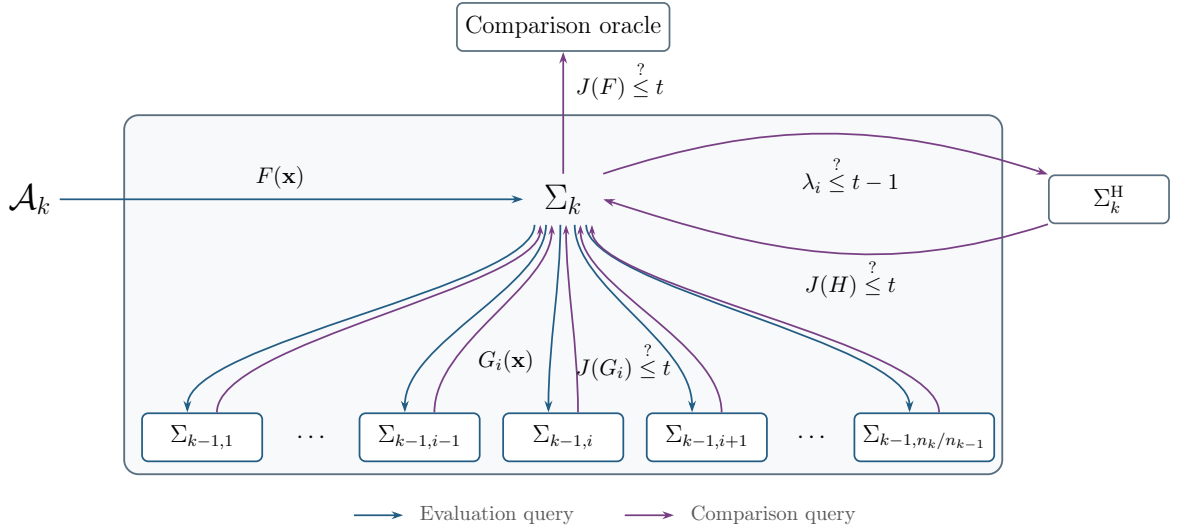
\begin{figure}[tbp]
    \centering
    \resizebox{\linewidth}{!}{\begingroup
\definecolor{qeval}{RGB}{35,93,132}
\definecolor{qcomp}{RGB}{120,63,132}
\definecolor{qframe}{RGB}{90,111,127}
\definecolor{qpale}{RGB}{247,249,251}
\begin{tikzpicture}[
  x=.7cm,y=.7cm,font=\normalsize,
  box/.style={draw=qframe,line width=.85pt,rounded corners=3pt,
    fill=white,minimum width=2.05cm,minimum height=.82cm,align=center},
  child/.style={box,draw=qeval},
  parent/.style={font=\Large,inner sep=5pt,minimum width=1.35cm,minimum height=.85cm},
  eval/.style={-{Stealth[length=2mm,width=1.35mm]},draw=qeval,line width=.85pt},
  comp/.style={-{Stealth[length=2mm,width=1.35mm]},draw=qcomp,line width=.85pt},
  return/.style={comp,-{Stealth[length=1.7mm,width=1.1mm]}},
  lab/.style={fill=none,draw=none,inner sep=2pt,font=\small},
  note/.style={font=\footnotesize,text=black!65},
  frame/.style={draw=qframe,line width=.9pt,rounded corners=7pt,fill=qpale}
]
\path[frame] (-10.7,-.5) rectangle (10.7,8.25);
\node[parent] (parent) at (0,6.2) {$\Sigma_k$};

\node[font=\Large] (algorithm) at (-13,6.2) {$\mathcal A_k$};
\node[box] (herringbone) at (13.3,6.2) {$\Sigma_k^{\mathrm H}$};
\node[box] (oracle) at (0,10.4) {Comparison oracle};

\node[child] (first)    at (-8.8,.4) {$\Sigma_{k-1,1}$};
\node[child] (previous) at (-3.5,.4) {$\Sigma_{k-1,i-1}$};
\node[child] (middle)   at (0,.4) {$\Sigma_{k-1,i}$};
\node[child] (next)     at (3.5,.4) {$\Sigma_{k-1,i+1}$};
\node[child] (last)     at (8.8,.4) {$\Sigma_{k-1,n_k/n_{k-1}}$};
\node at (-6.1,.4) {$\cdots$};
\node at (6.1,.4) {$\cdots$};

\foreach \child/\leftport/\rightport in {
  first/-.70/-.56,
  previous/-.42/-.28,
  middle/-.07/.07,
  next/.28/.42,
  last/.56/.70
}{
  \draw[eval] ($(parent.south)+(\leftport,0)$)
    .. controls +(0,-1.5) and +(0,1.7) ..
    ($(\child.north)+(-.36,0)$);
  \draw[return] ($(\child.north)+(.36,0)$)
    .. controls +(0,1.7) and +(0,-1.5) ..
    ($(parent.south)+(\rightport,0)$);
}

\node[lab,anchor=east] at (-.6,2.25) {$G_i(\mathbf x)$};
\node[lab,anchor=west] at (.2,2.25) {$J(G_i)\overset{?}{\leq}t$};

\draw[eval] (algorithm.east) -- (parent.west);
\node[lab,anchor=south] at (-6.9,6.35) {$F(\mathbf x)$};

\draw[comp] (parent.north east) to[out=18,in=162] (herringbone.north west);
\node[lab,anchor=north] at (7,7.35) {$\lambda_i\overset{?}{\leq}t-1$};
\draw[comp] (herringbone.south west) to[out=198,in=-18] (parent.east);
\node[lab,anchor=north] at (7,4.85) {$J(H)\overset{?}{\leq}t$};

\draw[comp] (parent.north) -- (oracle.south);
\node[lab,anchor=west] at (.2,9.1) {$J(F)\overset{?}{\leq}t$};

\draw[eval] (-5.1,-1.5) -- (-3.9,-1.5);
\node[note,anchor=west] at (-3.7,-1.5) {Evaluation query};
\draw[comp] (1.5,-1.5) -- (2.7,-1.5);
\node[note,anchor=west] at (2.9,-1.5) {Comparison query};
\end{tikzpicture}
\endgroup}
    \caption{Schematic diagram showing interactions between different agents discussed in the proof. Arrows show queries made in the direction of the arrow.}
    \label{fig:query_interfaces_five_enclosed}
\end{figure}

Now it remains to verify the consistency and query bound of $\Sigma_k$.

The query bound follows from the bound on the simulators $\Sigma_{k-1,i}$ and $\Sigma_k^H$:
After $Q$ queries made by $\AA_k$ and forwarded to the subsimulators as at most $3Q$ evaluation queries, they have made at most $O_k\left( \frac{Q}{\log^k n_{k-1}} \right)$ local queries $J(G_i) \overset{?}{\leq } t$ by the induction hypothesis. These local queries were then forwarded to $\Sigma_k^H$ as sublabel comparison queries. 
By Lemma \ref{lem:herringbone-simulation}, $\Sigma_k^H$ answered these queries while making at most $O_k\left( \frac{Q}{\log^{k} n_{k-1} \log m_k} \right)= O_k\left( \frac{Q}{\log^{k+1} n_{k} } \right)$ label comparison queries that were forwarded to $\OO_k$ as local queries by $\Sigma_k$. This verifies the bound on the local queries of $\Sigma_k$.

For consistency, we will choose component functions with $J(G_i) = \lambda_i +1$ for each $i$. These are admissible component labels, since $1\leq\lambda_i+1\leq4m_k^2+2m_k+2\leq m_{k-1}$. By Lemma \ref{lem:sublabels}, the sublabel-sum condition gives the switching index $j=2y^*-2+1_{x^*>a_{y^*}}$, which satisfies $1\leq j<8m_k^2$, and hence
\begin{equation*}
    J(F)=1+\left\lfloor\frac{2y^*-2+1_{x^*>a_{y^*}}}{8m_k}\right\rfloor
    =\left\lceil\frac{y^*}{4m_k}\right\rceil=J(H).
\end{equation*}
    
We need that at any point of the algorithm, for any label value $J$ that is consistent with all local queries made so far, there exists a function $F$ with label $J$ that is consistent with all answers given by $\Sigma_k$ (with the help of its subsimulators $\Sigma_{k-1,i}$ and $\Sigma_k^H$) so far.

Every comparison query made by $\Sigma_k^H$ was forwarded to the comparison oracle as a query about $J(F)$ with the same threshold, and received the same answer. Thus $J$ is consistent with all comparison answers received by $\Sigma_k^H$. By Lemma \ref{lem:herringbone-simulation}, there is a function $H$ with $J(H)=J$ that is consistent with all answers to sublabel queries $\lambda_i \overset{?}{\leq} t$ so far. Fix $\lambda_i=\lambda_i(H)$ for each $i$. 

For each $i$, every query $J(G_i)\overset{?}{\leq}t$ made by $\Sigma_{k-1,i}$ was answered by forwarding the query $\lambda_i\overset{?}{\leq}t-1$ to $\Sigma_k^H$. Therefore, $\lambda_i+1$ is consistent with all label comparison answers received by $\Sigma_{k-1,i}$.
By the induction hypothesis, there exists a function
$G_i'\in\FF_{k-1}$ with label $\lambda_i+1$ consistent with all evaluation answers returned by $\Sigma_{k-1,i}$. Choose such a function for each $i$, define the modified functions $G_i$ by (\ref{eq:Gdef}), and construct $F$ as in (\ref{eq:Fdef}).

The sublabel-sum condition for $H$ gives
\begin{equation*}
    J(G_i')+J(G_{i+1}') \geq 4m_k^2+2m_k+3 \quad\Longleftrightarrow\quad \lambda_i+\lambda_{i+1} \geq 4m_k^2+2m_k+1 \quad\Longleftrightarrow\quad i\leq j,
\end{equation*}
where $j$ is the switching index defined above. 
Thus condition (\ref{eq:Fcond2}) holds, so $F\in\FF_k$.
The label identity established above gives $J(F)=J(H)=J$.

Finally, consider any previous query $F(X,Y,Z)$.
Its target indices depend only on $(X,Y,Z)$, and each chosen component function reproduces the corresponding component answers used by $\Sigma_k$.
Consequently, $F(X,Y,Z)$ equals the answer previously returned by $\Sigma_k$. This proves consistency.

\bigskip

\bigskip

\textbf{Disclaimer of AI use and relation to the parallel results.}
ChatGPT-6 Astra was used to assist in developing the lower-bound argument, working out the technical details of the manuscript, preparing the figures, and checking the arguments.
The author takes full responsibility for the content of this paper.
In September 2026, the author independently obtained a deterministic lower bound of the same asymptotic form as those of Chen, Gao, Li, Rubinstein, and Yannakakis~\cite{CGLRY26} and Liu, Wang~\cite{LW26}.
The author subsequently strengthened this bound to the one presented in this paper.

\textbf{Acknowledgements.} The author would like to thank Dömötör Pálvölgyi for helpful discussions.

\bibliographystyle{plain}
\bibliography{bibliography}

\end{document}